%% file: arxiv_rev.tex
\documentclass{article}
\usepackage[utf8]{inputenc}
\usepackage[dvipdfmx]{graphicx}
\usepackage{fullpage}
\usepackage{amsthm}
\usepackage{amsmath}
\usepackage{amssymb}
\usepackage{mathtools}
\usepackage[colorlinks,citecolor=blue,bookmarks=true,linktocpage]{hyperref}
\usepackage{cleveref}
\usepackage[backend=biber,style=numeric,sorting=nyt,url=false,arxiv=abs,maxbibnames=100,isbn=false]{biblatex}
\usepackage{url}

\input{macro}

\begin{document}

\title{Reconfiguring (non-spanning) arborescences\thanks{
This is a post-peer-review, pre-copyedit version of an article published in Theoretical Computer Science. The final authenticated version is available online at \protect\url{https://doi.org/10.1016/j.tcs.2022.12.007}.
A preliminary version of this paper~\cite{IIKNOW:C0C00N:2021} has appeared in the proceedings of the 27th International Computing and Combinatorics Conference (COCOON 2021).
Partially supported by JSPS KAKENHI grant numbers
JP18H04091,
JP18K11168,
JP18K11169,
JP19K11814,
JP20K19742,
JP19K20350,
JP19J21000,
JP20H05793,
JP20H05795,
JP20K23323,
JP21K11752,
and JST, CREST Grant Number
JPMJCR18K3,
Japan.
}}

\author{Takehiro Ito\thanks{Graduate School of Information Sciences, Tohoku University, Japan. Email: \texttt{takehiro@tohoku.ac.jp}} \and
Yuni Iwamasa\thanks{Graduate School of Informatics, Kyoto University, Japan. Email: \texttt{iwamasa@i.kyoto-u.ac.jp}} \and
Yasuaki Kobayashi\thanks{Graduate School of Information Science and Technology, Hokkaido University. Email: \texttt{koba@ist.hokudai.ac.jp}} \and
Yu Nakahata\thanks{Division of Information Science, Nara Institute of Science and Technology, Japan. Email: \texttt{yu.nakahata@is.naist.jp}} \and
Yota Otachi\thanks{Graduate School of Informatics, Nagoya University, Japan. Email: \texttt{otachi@nagoya-u.jp}} \and
Kunihiro Wasa\thanks{Faculty of Science and Engineering, Hosei University, Japan. Email: \texttt{wasa@hosei.ac.jp}}
}

\maketitle

\begin{abstract}
In this paper, we investigate the computational complexity of subgraph reconfiguration problems in directed graphs.
More specifically, we focus on the problem of reconfiguring arborescences in a digraph, where an arborescence is a directed graph such that its underlying undirected graph forms a tree and
all vertices have in-degree at most 1.
Given two arborescences in a digraph, the goal of the problem is to determine whether there is a (reconfiguration) sequence of arborescences between the given arborescences such that each arborescence in the sequence can be obtained from the previous one by removing an arc and then adding another arc.
We show that this problem can be solved in polynomial time, whereas the problem is PSPACE-complete when we restrict arborescences in a reconfiguration sequence to directed paths or relax to directed acyclic graphs.
We also show that there is a polynomial-time algorithm for finding a shortest reconfiguration sequence between two spanning arborescences.
\end{abstract}

\section{Introduction}\label{sec:intro}
Let $\Pi$ be a graph structure property. 
For a graph $G$, we denote by $\mathcal{S}_{\Pi}(G)$ the set of all subgraphs of $G$ that satisfy $\Pi$ and have the same number of edges.
In this paper, we study the reachability of the solution space formed by $\mathcal{S}_{\Pi}(G)$, where two subgraphs $H$ and $H'$ in $\mathcal{S}_{\Pi}(G)$ are \emph{adjacent} in the solution space if and only if they can be obtained from each other by swapping a pair of edges, that is, $|E(H) \setminus E(H')| = |E(H') \setminus E(H)| = 1$.
Our target is to decide whether there is a (reconfiguration) sequence of adjacent subgraphs in $\mathcal{S}_{\Pi}(G)$ between two given subgraphs $\source{H}$ and $\sink{H}$ in $\mathcal{S}_{\Pi}(G)$.
To avoid confusion, we sometimes call the problem the \emph{reachability variant}, because we will study the shortest sequence variant later.

The problem has been studied for several graph structure properties $\Pi$ (on undirected graphs), although most of the related results appear under the name of the property $\Pi$ under consideration. 
For example, {\sc Spanning Tree Reconfiguration} can be seen as the problem when $\Pi$ is the property of being a spanning tree.
Every instance of this problem is a yes-instance because the set of spanning trees is the family of bases of a matroid~\cite{IDHPSUU11}.
Ito et~al.~\cite{IDHPSUU11} showed that when $\Pi$ is the property of being a matching the problem is solvable in polynomial time, and M\"uhlenthaler~\cite{Muh15} extended the result to degree-constrained subgraphs.
Hanaka et~al.~\cite{HIMMNSSV20} introduced the framework of subgraph reconfiguration problems, and studied the problem for several properties $\Pi$, including trees and paths.
In particular, they showed that when $\Pi$ is the property of being a tree, every instance of the problem is a yes-instance unless two input trees have different numbers of edges.
Motivated by applications in motion planning, Biasi and Ophelders~\cite{BO18}, Demaine~et al.~\cite{DEHJLUU19}, and Gupta et~al.~\cite{GSZ20} studied some variants of reconfiguring undirected paths.
These variants are shown to be PSPACE-complete in general, while they are fixed-parameter tractable when parameterized by the length of input paths.

In contrast to various results for undirected graphs, the problem was not well-studied for directed graphs.
In this paper, we investigate the complexity of subgraph reconfiguration problems on directed graphs.
We mainly study the problem when $\Pi$ is the property of being an arborescence, where an arborescence is a directed graph such that its underlying undirected graph forms a tree and every vertex except for exactly one vertex has in-degree $1$.
Note that two (directed) subgraphs in $\mathcal{S}_{\Pi}(G)$ are \emph{adjacent} if and only if they can be obtained from each other by swapping a pair of arcs (instead of a pair of edges). 
We refer to this problem as {\sc Arborescence Reconfiguration}.
(Formal definitions will be given in \Cref{sec:preli}.)
Interestingly, {\sc Arborescence Reconfiguration} has no-instances as shown in \Cref{fig:no-instance}, in contrast to the fact that any two undirected trees are reconfigurable as long as they have the same number of edges~\cite{HIMMNSSV20}.
Nonetheless we give the following theorem, as our main result.
\begin{theorem}\label{thm:dtr}
    Let $G = (V,A)$ be a directed graph. 
    \textsc{Arborescence Reconfiguration} can be solved in $O(|V||A|)$ time.
    Moreover, if the answer is affirmative, we can construct a reconfiguration sequence between two given arborescences of length $O(|V|^2)$ in polynomial time.
\end{theorem}

We further investigate the problem for specific arborescences.
By the definition, an arborescence has a unique vertex $r$ whose in-degree is $0$.
We call $r$ the \emph{root} of the arborescence, and call an arborescence with root $r$ an \emph{$r$-arborescence}. 
We will show that any two $r$-arborescences are reconfigurable when $\Pi$ is the property of being an $r$-arborescence with a prescribed vertex $r$.
This result gives an interesting contrast to {\sc Arborescence Reconfiguration} (recall the no-instance in \Cref{fig:no-instance}), and will play an important role in our proof of Theorem~\ref{thm:dtr}. 

\begin{figure}[t]
    \centering
    \includegraphics[width=0.2 \linewidth]{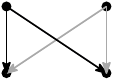}
    \caption{There is no reconfiguration sequence between the black and gray arborescences.}
    \label{fig:no-instance}
\end{figure}

We also consider the cases where $\Pi$ is the property of being a directed path, a directed acyclic graph (DAG), or a strongly connected graph.
Formal definitions will be given in \Cref{sec:hard}.
For these cases, we show negative results.
Our results are summarized in \Cref{tab:result}.

\begin{table}[t]
    \centering
    \caption{Summary of our results. For cases marked with `*', the answer is affirmative if and only if two given subgraphs have the same number of arcs.}
    \label{tab:result}
\begin{tabular}{|c|c|c|}
\hline
Property $\Pi$                      & Reachability variant & Shortest sequence variant \\ \hline\hline
arborescence              & P           & open             \\ \hline
$r$-arborescence          & always yes*          & open             \\ \hline
spanning arborescence     & always yes           &
\begin{tabular}{c}
P\\ (with $|A(\source{T}) \setminus A(\sink{T})|$ steps)
\end{tabular}\\ \hline
spanning $r$-arborescence & always yes           & 
\begin{tabular}{c}P\\ (with $|A(\source{T}) \setminus A(\sink{T})|$ steps)
\end{tabular}
\\ \hline
directed path              & PSPACE-complete      & -                \\ \hline
directed acyclic graph     & PSPACE-complete      & -                \\ \hline
strongly connected vertex set & PSPACE-complete              & -             \\ \hline
strongly connected arc set & NP-hard              & -             \\ \hline
\end{tabular}
\end{table}

In this paper, we also study the shortest sequence variant, which computes the shortest sequence of a reconfiguration sequence between two given subgraphs in $\mathcal{S}_{\Pi}(G)$. 
In particular, {\sc Spanning Arborescence Shortest Reconfiguration} is the shortest sequence variant when $\Pi$ is the property of being a spanning arborescence.
We will prove the following theorem, by constructing a reconfiguration sequence between two spanning arborescences $\source{T}$ and $\sink{T}$ of length $|A(\source{T}) \setminus A(\sink{T})| = |A(\sink{T}) \setminus A(\source{T})|$.
\begin{theorem}\label{thm:sdstr}
    {\sc Spanning Arborescence Shortest Reconfiguration} can be solved in polynomial time.
\end{theorem}

When $\Pi$ is the property of being a spanning arborescence, the reachability variant can be seen as a special case of {\sc Matroid Intersection Reconfiguration}
for a graphic matroid and (a truncation of) a partition matroid.
Here, given two matroids and their two common bases $\source{B}$ and $\sink{B}$, {\sc Matroid Intersection Reconfiguration} asks to determine if there is a reconfiguration sequence of common bases between $\source{B}$ and $\sink{B}$;
see \cite{book/Oxley11} for matroids.
It is shown in \cite{IDHPSUU11} that {\sc Maximum Bipartite Matching Reconfiguration} is solvable in polynomial time.
While this problem can be seen as {\sc Matroid Intersection Reconfiguration} for two (truncations of) partition matroids, the complexity of {\sc Matroid Intersection Reconfiguration} remains open. 
\Cref{thm:sdstr} provides a new tractable class of {\sc Matroid Intersection Reconfiguration}, particularly, its shortest sequence version.

\paragraph{Paper organization}
In \Cref{sec:preli}, we define some notation and terminology used in this paper.
\Cref{sec:vari} gives some positive results for variants of {\sc Arborescence Reconfiguration}.
Based on one of these results in \Cref{sec:vari}, we develop a polynomial-time algorithm for {\sc Arborescence Reconfiguration} in \Cref{sec:dtr}, proving \Cref{thm:dtr}.
Then, we give some negative results in \Cref{sec:hard}, and then conclude our paper in the last section.

\section{Preliminaries}\label{sec:preli}

Let $G = (V, A)$ be a directed graph.
We denote by $V(G)$ and $A(G)$ the vertex and arc sets of $G$, respectively.
Let $e = (u, v)$ be an arc of $G$.
We say that $e$ is \emph{directed from $u$} or \emph{directed to $v$}.
The vertex $u$ (resp.\ $v$) is called the \emph{tail} (resp.\ \emph{head}) of $e$.
For each $v \in V$, we denote by $N^+_G(v)$ the set of out-neighbors of $v$ in $G$, i.e., $N^+_G(v) = \{w \in V : (v, w) \in A\}$.
The \emph{in-degree} (resp.\ \emph{out-degree}) of $v$ is the number of arcs directed to $v$ (resp.\ directed from $v$) in $G$.
For a subset $X \subseteq V$, the subgraph of $G$ induced by $X$ is denoted by $G[X]$.
For an arc $(u, v) \in G$ and a subgraph $H$ of $G$,
we denote by $H + (u,v)$ and $H - (u,v)$ the directed graphs obtained from $H$ by adding $(u, v)$ and by removing $(u, v)$, respectively.

An \emph{arborescence} $T$ is a directed graph such that its underlying undirected graph forms a tree and every vertex except for a vertex $r \in V(T)$ has in-degree exactly $1$.
The unique vertex $r$ of in-degree $0$ is called the \emph{root} of $T$, and $T$ is called an \emph{$r$-arborescence}.
An ($r$-)arborescence $T$ that is a subgraph of $G$ is said to be \emph{spanning} if $V(T) = V(G)$.\footnote{Let us note that the term ``arborescences'' may be used in different meanings, where the definition of ``arborescence'' sometimes requires to be a spanning subgraph.
In our definition, arborescences are not required to be spanning subgraphs.
We call arborescences that are spanning subgraphs \emph{spanning arborescences}.
}
A directed graph consisting of a disjoint union of arborescences is called a \emph{branching} or an \emph{$R$-branching}, where $R$ is the set of roots of its (weakly) connected components.
An arc in an arborescence $T$ is called a \emph{leaf arc} if the out-degree of its head is $0$ in $T$. 
A \emph{directed path} is an arborescence that has at most one leaf arc.

Let $\Pi$ be a graph structure property. 
For a graph $G$, we denote by $\mathcal{S}_{\Pi}(G)$ the set of all subgraphs of $G$ that satisfy $\Pi$.
Let $H$ and $H'$ be two subgraphs in $\mathcal{S}_{\Pi}(G)$ that have the same size.
A sequence $\langle H_0, H_1, \ldots, H_\ell\rangle$ of subgraphs in $\mathcal{S}_{\Pi}(G)$ is called a \emph{reconfiguration sequence between $H$ and $H'$} if $H_0 = H$, $H_\ell = H'$, and $|A(H_i) \setminus A(H_{i + 1})| = |A(H_{i + 1}) \setminus A(H_i)| = 1$ for all $i$, $0 \le i < \ell$.
In other words, $H_{i + 1}$ can be obtained by removing an arc from $H_i$ and then adding another arc to it for each $i$ with $0 \le i < \ell$.
We call $\ell$ the \emph{length} of the reconfiguration sequence.
If there is a reconfiguration sequence between $H$ and $H'$, we say that $H$ is \emph{reconfigurable} from $H'$.
Note that any reconfiguration sequence is reversible: 
$H'$ is reconfigurable from $H$ if and only if $H$ is reconfigurable from $H'$.
For simplicity, we assume without loss of generality that all subgraphs in $\mathcal{S}_{\Pi}(G)$ have the same size; otherwise they are not reconfigurable. 

\section{Always Reconfigurable Cases}\label{sec:vari}

In this section, we show that every instance of the reachability variant is a yes-instance for some graph properties $\Pi$.

\subsection{Branchings}

Let $\mathcal S \subseteq 2^U$ be a collection of subsets of a finite set $U$.
Suppose that every set in $\mathcal S$ has the same cardinality.
We say that $\mathcal S$ satisfies the \emph{weak exchange property}\footnote{Note that our definition of weak exchange property is different from \emph{weak exchange axiom}, which is introduced for M-convex functions in Murota's book~\cite[p.\ 137]{Murota03}.} if for $S, S' \in \mathcal S$ with $S \neq S'$, there exist $e \in S \setminus S'$ and $e' \in S' \setminus S$ such that $S \setminus \{e\} \cup \{e'\} \in \mathcal S$.
This property is closely related to the 
exchange property of bases of matroids: Recall that if $\mathcal B$ is the collection of bases of a matroid, then for $B, B' \in \mathcal B$ with $B \neq B'$ and for $e \in B \setminus B'$, there is $e' \in B' \setminus B$ such that $B \setminus \{e\} \cup \{e'\} \in \mathcal B$.
The weak exchange property is not only a weaker version of the exchange property but also gives an important consequence for reconfiguration problems in the following sense, which can be easily observed.
\begin{lemma}\label{lem:WEP}
Let $\Pi$ be a graph structure property.
    All two subgraphs $H$ and $H'$ in $\mathcal{S}_{\Pi}(G)$ admit a reconfiguration sequence of length $|A(H') \setminus A(H)| = |A(H) \setminus A(H')|$
    if and only if
    $\mathcal{S}_{\Pi}(G)$ satisfies the weak exchange property.
\end{lemma}

Since the lower bound of the length of a reconfiguration sequence between $H$ and $H'$ is clearly $|A(H') \setminus A(H)| = |A(H) \setminus A(H')|$,
\Cref{lem:WEP} implies that,
if $\mathcal{S}_{\Pi}(G)$ satisfies the weak exchange property
then the shortest sequence variant can be solved in polynomial time for the property $\Pi$.

In this subsection, we show that $\mathcal{S}_{\Pi}(G)$ satisfies the weak exchange property for some graph structure properties $\Pi$.
We first show that, similar to the undirected case~\cite{IDHPSUU11}, the weak exchange property holds when $\Pi$ is the property of being a spanning arborescence.
\begin{theorem}
\label{thm:dspntree}
 $\mathcal{S}_{\Pi}(G)$ satisfies the weak exchange property when $\Pi$ is the property of being a spanning arborescence.
\end{theorem}
\begin{proof}
Let $T$ and $T'$ be arbitrary spanning arborescences in $G$ with $T \neq T'$.
Suppose first that $T$ and $T'$ have a common root $r$.
Let $e' = (u,v)$ be an arc in $T' \setminus T$ such that the path from $r$ to $u$ in $T'$ is contained in $T$.
Clearly, we have $v \neq r$.
Let $e$ be the unique arc directed to $v$ in $T$.
From the definition of $e$ and $e'$, we have $e \ne e'$.
Let $R = T + e' - e$.
Now in $T + e'$, the vertex $v$ is the only vertex that has two arcs ($e$ and $e'$) directed to it.
Thus, in $R$, no vertex has in-degree 2 or more.
Moreover, all vertices in $R$ are reachable from $r$:
the paths in $T$ that use $e$ are rerouted to use $e'$ in $R$,
and all other paths in $T$ still exist in $R$.
Since $|T| = |R|$, $R$ is a spanning arborescence in $G$.

Suppose next that $T$ and $T'$ have different roots $r$ and $r'$, respectively.
Let $e'$ be the unique arc in $T'$ directed to $r$, that is, $e' = (u,r)$ for some $u \in V$.
Let $P$ be the path from $r$ to $u$ in $T$.
Since $P + e'$ is a directed cycle, 
there is an arc $e = (v,w) \in P$ that does not belong to $T'$.
Let $R = T + e' - e$.
Observe that no vertex in $R$ has in-degree 2 or more
since it holds already in $P + e'$.
Observe also that all vertices in $R$ are reachable from $w$:
for the descendants of $w$ in $T$, $R$ contains the same path from $w$;
and for the other vertices, we first follow the path from $w$ to $u$ in $T$, 
use the arc $e' = (u,r)$, and then follow the path in $T$ from $r$.
Since $|T| = |R|$, $R$ is a spanning arborescence (rooted at $w$) in $G$.
\end{proof}

From the proof of \Cref{thm:dspntree}, we obtain the following corollary. 
\begin{corollary}\label{cor:r-span}
    $\mathcal{S}_{\Pi}(G)$ satisfies the weak exchange property when $\Pi$ is the property of being an spanning $r$-arborescence.
\end{corollary}

We then prove the following theorem, which implies that the shortest sequence variant is solvable in polynomial time when $\Pi$ is a branching.
\begin{theorem}
\label{theorem:dforest}
    $\mathcal{S}_{\Pi}(G)$ satisfies the weak exchange property when $\Pi$ is the property of being a branching.
\end{theorem}
\begin{proof}

Let $F$ and $F'$ be distinct branchings in $G$ with $|A(F)| = |A(F')|$.
We first consider the case where there is some arc $e' \in F' \setminus F$
such that the endpoints of $e'$ do not belong to the same (weakly) connected component of $F$, that is, either $e'$ connects two connected components of $F$ or at least one of the endpoints of $e'$ does not belong to $F$.
Now, we show that there is an arc $e \in F \setminus F'$ such that $F + e' - e$ is a branching of $G$. 
If $F + e'$ is a branching, then we can select any arc in $F \setminus F'$ as $e$.
Assume that $F + e'$ is not a branching.
By the assumption in this case, the underlying undirected graph of $F + e'$ contains no (undirected) cycle.
Thus there is a vertex of in-degree at least 2 in $F + e'$.
Since $F$ is a branching, only the head of $e'$, say $v$, can be such a vertex, and its in-degree is exactly 2.
As $e$, we select the other arc in $F + e'$ that has $v$ as its head.
Since $e' \in F' \setminus F$, this arc $e$ does not belong to $F'$.
Since $F + e' - e$ does not contain any cycle in the underlying graph nor any vertex of in-degree 2 or more,
it is a branching in $G$.

Next we consider the case where every arc $e' \in F' \setminus F$ has both endpoints in the same connected component of $F$.
Let $F_{1}, \dots, F_{c} \subseteq F$ be the connected components of $F$,
and let $F'_{1}, \dots, F'_{c} \subseteq F'$ be 
the subsets of $F'$ such that $F'_{i} = \{e' \in F' \mid e' \text{ has both endpoints in } F_{i}\}$.
We claim that $|A(F_{i})| = |A(F'_{i})|$.
To see this, observe that if $|A(F_{i})| < |A(F'_{i})|$ for some $i$,
then $F'_{i}$ is not an arborescence since $V(F'_{i}) \subseteq V(F_{i})$ and
$F_{i}$ is a spanning arborescence of the subgraph of $G$ induced by $V(F_{i})$.
This proves the claim as $|A(F)| = |A(F')|$.
Since both endpoints of every arc in $F'_i$ belong to $F_i$, we also have $V(F_i) = V(F'_i)$ for all $1 \le i \le c$. 
As $F \neq F_i$, there is a connected component $F_i$ in $F$ with $F_i \neq F'_i$ and by \Cref{thm:dspntree}, the theorem follows.
\end{proof}

By combining \Cref{lem:WEP} with \Cref{thm:dspntree},
\Cref{cor:r-span}, and \Cref{theorem:dforest}, we immediately obtain the following,
which particularly implies \Cref{thm:sdstr}.
\begin{theorem}
    Let $\Pi$ be one of the graph structure properties of being a spanning arborescence, a spanning $r$-arborescence, and a branching.
    Then, for any $H, H' \in \mathcal{S}_{\Pi}(G)$,
    there exists a reconfiguration sequence of length $|A(H') \setminus A(H)| = |A(H) \setminus A(H')|$.
\end{theorem}

As mentioned in \Cref{sec:intro}, when $\Pi$ is the property of being a spanning arborescence (or a branching), the reachability variant is a subclass of {\sc Matroid Intersection Reconfiguration}.
\Cref{thm:dspntree,theorem:dforest} give a new insight on matroid intersection in terms of the weak exchange property.

\begin{remark}
The family of all
$r$-arborescences in a graph is a typical example of greedoids; a set family $\mathcal S \subseteq 2^U$ of a finite set $U$
is called a \emph{greedoid}~\cite{KLS1991} if it satisfies that $\emptyset \in \mathcal S$ and for any $X, Y \in \mathcal S$ with $|X| < |Y|$, there is $y \in Y \setminus X$ such that $X \cup \{y\} \in \mathcal S$.
Lov\'{a}sz~\cite{lovasz1977homology} showed that,
if a graph is $2$-connected
then, for any two spanning $r$-arborescences $T$ and $T'$,
there exists a reconfiguration sequence $\langle T =: T_0, T_1, \ldots, T_\ell := T'\rangle$
such that, for each $i$, the arcs $e_{i-1} \in T_{i-1} \setminus T_i$ and $e_i \in T_i \setminus T_{i-1}$
are leaves in $T_{i-1}$ and in $T_i$, respectively.
See \cite[Theorem 2.11]{KLS1991} for its generalization to greedoids.

Lov\'{a}sz's result and ours are incomparable in the following sense.
Indeed, Lov\'{a}sz dealt with a more restricted reconfiguration rule (the arcs $e_{i-1} \in T_{i-1} \setminus T_i$ and $e_i \in T_i \setminus T_{i-1}$
must be leaves in $T_{i-1}$ and in $T_i$, respectively) than ours, but a digraph is required to be 2-connected for the reconfigurability of any two spanning $r$-arborescences. Furthermore, the length of any reconfiguration sequence between $T$ and $T'$ can be strictly larger than the lower bound $|A(T) \setminus A(T')|$ under Lov\'{a}sz's rule.
On the other hand, in our setting, the $2$-connectivity of a digraph is not required for the reconfigurability,
and there is always a reconfiguration sequence of length $|A(T) \setminus A(T')|$ between spanning $r$-arborescences $T$ and $T'$.
\end{remark}

\subsection{Branchings with fixed roots}

In this subsection, we consider the case where the property $\Pi$ is the property of being an $r$-arborescence for a fixed vertex $r$. 
Then, every instance of the reachability variant is a yes-instance and admits a reconfiguration sequence of linear length. 
More precisely, we prove the following theorem.

\begin{theorem}\label{thm:fixed-root}
    For every pair of $r$-arborescences $T$ and $T'$ in $G$ with $|A(T)| = |A(T')| = k$, there is a reconfiguration sequence $\langle T = T_0, T_1, \ldots, T_\ell = T' \rangle$ such that all intermediate arborescences have the same root $r$.
    Moreover, the length $\ell$ of the reconfiguration sequence is at most $k$.
\end{theorem}
\begin{proof}
We say that an arc $e$ in an arborescence $T''$ is \emph{fixed} (with respect to $T'$) if the directed path from $r$ to the head of $e$ in $T''$ appears in $T'$.
An arc is \emph{unfixed} if it is not fixed.
Let $h$ be the number of unfixed arcs in $T$.
We prove that there is a reconfiguration sequence between $T$ and $T'$ of length at most $h$ by induction on $h$. 
If $h=0$, then we have $T = T'$.
In the following, we assume that $h \ge 1$ and that for every $r$-arborescence $T''$ that has $k$ arcs and contains fewer than $h$ unfixed arcs with respect to $T'$, there is a reconfiguration sequence from $T''$ to $T'$ of length at most $h - 1$.

Let $e = (u,v)$ be an arc in $T'$ such that $e$ is not included in $T$ 
but all other arcs in the path $P$ from $r$ to the tail of $e$ in $T'$ are included in $T$.
Such an arc exists since $T \ne T'$ and they share the root $r$.
Note that all arcs in $P$ are fixed.

Assume for now that there is an unfixed arc $f$ in $T$ such that $T'' \coloneqq T + e - f$ is an arborescence in $G$.
Note that $T''$ is still rooted at $r$ since $e$ is an arc of an arborescence rooted at $r$.
Observe that arc $e$ is fixed in $T''$ as both $T''$ and $T'$ contain the path $P$ and that the fixed arcs of $T$ remain fixed in $T''$ since we only removed the unfixed arc $f$.
Thus $T''$ has fewer than $h$ unfixed arcs. 
By the induction hypothesis, there is a reconfiguration sequence from $T''$ to $T'$ of length at most $h - 1$, and thus $T'$ is reconfigurable from $T$ as $|A(T) \setminus A(T'')| = |A(T'') \setminus A(T)| = 1$.
Therefore, it suffices to find such an arc $f$.

If the head $v$ of $e$ is included in $T$,
then we set $f$ to the arc directed to $v$ in $T$.
Then $f$ is unfixed since $T'$ cannot contain it
and $T + e - f$ is an arborescence obtained from $T$ by changing the parent of $v$ to $u$.
Otherwise, $v$ is not included in $T$,
then we set $f$ to an unfixed leaf arc of $T$, which exists since $h \ge 1$.
Since $T + e$ is an arborescence
and $f$ is a leaf arc of $T + e$ as well, $T + e - f$ is an arborescence.
\end{proof}

This result can be extended to $R$-branchings.

\begin{theorem}
    For every pair of $R$-branchings $F$ and $F'$ in $G$ with $|A(F)| = |A(F')| = k$, there is a reconfiguration sequence $\langle F = F_0, F_1, \ldots, F_\ell = F' \rangle$ such that all intermediate forests are $R$-directed.
    Moreover, the length $\ell$ of the reconfiguration sequence is at most $k$.
\end{theorem}

\begin{proof}
    Since every arc between vertices in $R$ cannot belong to any $R$-branching, we assume that there are no arcs between them.
    Let $G'$ be a directed multigraph obtained from $G$ by identifying vertices in $R$ into a single vertex $r$. 
    Observe that a set $X \subseteq E$ forms an $R$-branching in $G$ if and only if $X$ is an $r$-arborescence of $G'$.
    By~\Cref{thm:fixed-root}, the theorem follows.
\end{proof}

\section{Algorithm for {\sc Arborescence Reconfiguration}}\label{sec:dtr}

This section is devoted to proving our main result, \Cref{thm:dtr}, which is a polynomial-time algorithm for {\sc Arborescence Reconfiguration}.
Recall that there are no-instances for the problem, as shown in \Cref{fig:no-instance}. 

The idea of our algorithm is as follows.
Let $G = (V, A)$ be a directed graph, and let $k$ be a positive integer.
For each $v \in V$, we denote by $\mathcal T(v)$ the collection of all $v$-arborescences $T$ in $G$ with $|A(T)| = k$.
By~\Cref{thm:fixed-root}, there is a reconfiguration sequence between any pair of $v$-arborescences in $\mathcal T(v)$ such that all internal arborescences in the sequence belong to $\mathcal T(v)$.
This enables us to ``compress'' all arborescences in $\mathcal T(v)$ into a single representative for each $v \in V$, and it suffices to seek the reachability in the ``compressed'' solution space. 
In the rest of this section, when we refer to reconfiguration sequences, every subgraph in these sequences are arborescences with $k$ arcs.

    Let $u$ and $v$ be distinct vertices in $G$, and let $T \in \mathcal T(u)$ and $T' \in \mathcal T(v)$.
\begin{lemma}\label{lem:adjacent}
    Suppose that $G$ has an arc $(u, v)$ or $(v, u)$.
    Then, there is a reconfiguration sequence between $T$ and $T'$.
\end{lemma}
\begin{proof}
Assume without loss of generality that $G$ has an arc $(u, v)$.
Since $v$ is the root of $T'$, we have $(u, v) \notin T'$.
If $u \notin V(T')$, the subgraph $T''$ obtained from $T' + (u, v)$ by removing arbitrary one of the leaf arcs is an arborescence in $\mathcal T(u)$.
Thus, by \Cref{thm:fixed-root}, there is a reconfiguration sequence between $T$ and $T''$ and then we are done in this case.
Otherwise, $T' + (u, v)$ has a directed cycle passing through $(u, v)$.
Then, the graph obtained from $T + (u, v)$ by removing the arc directed to $u$ in the cycle is an arborescence in $\mathcal T(u)$.
Again, by \Cref{thm:fixed-root}, the lemma follows.
\end{proof}

By inductively applying \Cref{thm:fixed-root} and this lemma, we have the following corollary.

\begin{corollary}\label{cor:suf:path}
    Suppose that $G$ has a directed path from $u$ to $v$ or from $v$ to $u$. Then, there is a reconfiguration sequence between $T$ and $T'$.
\end{corollary}

\begin{lemma}\label{lem:suf:common}
    If there is a vertex
    $w \in N^+_G(u) \cap N^+_G(v)$
    such that $G[V \setminus \{u, v\}]$ has a $w$-arborescence of size $k - 1$,
    then there is a reconfiguration sequence between $T$ and $T'$.
\end{lemma}
\begin{proof}
    Let $T''$ be an arborescence in $G[V \setminus \{u, v\}]$ that has $k - 1$ arcs and root $w \in N^+_G(u) \cap N^+_G(v)$.
    Since $T'' + (u, w)$ and $T'' + (v, w)$ are arborescences that belong to $\mathcal T(u)$ and $\mathcal T(v)$, respectively,
    by \Cref{thm:fixed-root}, there are reconfiguration sequences between $T$ and $T'' + (u, w)$ and between $T'' + (v, w)$ and $T'$.
    As $T'' + (v, w)$ is reconfigurable from $T'' + (u, w)$, concatenating these sequences yields a reconfiguration sequence between $T$ and $T'$.
\end{proof}

The above corollary and lemma give sufficient conditions for finding a reconfiguration sequence between $T$ and $T'$.
The following lemma ensures that these conditions are also necessary conditions for a ``single step''.

\begin{lemma}\label{lem:nece}
    Suppose that $|A(T) \setminus A(T')| = |A(T') \setminus A(T)| = 1$.
    Then, at least one of the following conditions hold: (1) $G$ has a directed path from $u$ to $v$ or from $v$ to $u$ or
    (2) there is
    $w \in N^+_G(u) \cap N^+_G(v)$
    such that $G[V \setminus \{u, v\}]$ has a $w$-arborescence of size $k - 1$.
\end{lemma}
\begin{proof}
    Suppose that $v \in V(T)$.
    Then, there is a directed path $P$ from $u$ to $v$ in $T$ and hence we are done.
    Symmetrically, the lemma follows when $u \in V(T')$.
    Thus, we assume that $v \notin V(T)$ and $u \notin V(T')$.
    This assumption implies that there is a unique arc $e$ directed from $u$ in $T$ as otherwise we have $|A(T) \setminus A(T')| \ge 2$.
    Also, there is a unique arc $e'$ directed from $v$ in $T'$.
    By the fact that $|A(T) \setminus A(T')| = |A(T') \setminus A(T)| = 1$, $T - e$ ($=T' - e'$) must be an arborescence with root $w \in N^+_G(u) \cap N^+_G(v)$ that has $k - 1$ arcs in $G[V \setminus \{u, v\}]$.
\end{proof}

To find a reconfiguration sequence between $\source{T}$ and $\sink{T}$, we construct an auxiliary graph $\mathcal G$ as follows.
We assume that $G$ is (weakly) connected. 
For each $v \in V$, $\mathcal G$ contains a vertex $v$ if $G$ has a $v$-arborescence of size $k$.
For each pair of distinct $u$ and $v$ in $V(\mathcal G)$, we add an (undirected) edge between them if (1) $G$ has a directed path from $u$ to $v$ or from $v$ to $u$; or (2) there is a vertex
$w \in N^+_G(u) \cap N^+_G(v)$
such that $G[V \setminus \{u, v\}]$ has a $w$-arborescence of size $k - 1$.
The graph $\mathcal G$ can be constructed in $O(|V||A|)$ time.
Our algorithm simply finds a path in $\mathcal G$ between the two roots of given arborescences $\source{T}$ and $\sink{T}$.
The correctness of the algorithm immediately follows from the following lemma, which also proves the first part of \Cref{thm:dtr}.

\begin{lemma}
    Let $\source{T}$ and $\sink{T}$ be arborescences in $G$ with $|A(\source{T})| = |A(\sink{T})| = k$ whose roots are $\source{r}$ and $\sink{r}$, respectively.
    Then, there is a path between $\source{r}$ and $\sink{r}$ in $\mathcal G$ if and only if there is a reconfiguration sequence between $\source{T}$ and $\sink{T}$.
\end{lemma}
\begin{proof}
    We first show the forward implication.
    Suppose that there is a path $\mathcal P$ between $\source{r}$ and $\sink{r}$ in $\mathcal G$.
    By \Cref{cor:suf:path} and \Cref{lem:suf:common} there is a reconfiguration sequence between $\source{T}$ and $\sink{T}$ that can be constructed along the path $\mathcal P$.
    
    For the converse implication, suppose that there is a reconfiguration sequence between $\source{T}$ and $\sink{T}$.
    Let $T$ and $T'$ be two arborescences that appear consecutively in the sequence.
    We claim that either $T$ and $T'$ have a common root or the roots of $T$ and $T'$ are adjacent in $\mathcal G$.
    If $T$ and $T'$ have a common root, the claim obviously holds.
    Suppose otherwise. Let $u$ and $v$ be the roots of $T$ and $T'$, respectively.
    By \Cref{lem:nece}, at least one of the conditions (1) and (2) holds, implying that $u$ and $v$ are adjacent in $\mathcal G$.
\end{proof}

We can modify our algorithm to find an actual reconfiguration sequence of length $O(|V|^2)$ if the answer is affirmative.
Let $P = (\source{r} = r_0, r_1, \dots, r_\ell = \sink{r})$ be a path between $\source{r}$ and $\sink{r}$.
We construct a reconfiguration sequence from $\source{T}$ to $\sink{T}$ by moving the roots from $\source{r}$ to $\sink{r}$ along $P$.
For each $0 \le i < \ell$, let $T_i$ and $T_{i+1}$ be an arborescence rooted at $r_i$ and $r_{i+1}$.
If the edge $\{r_i, r_{i+1}\}$ in $\mathcal{G}$ is type (1), there exists a reconfiguration sequence from $T_i$ to $T_{i+1}$ of length $O(|V|)$ by \Cref{lem:adjacent} and \Cref{cor:suf:path}.
If the edge $\{r_i, r_{i+1}\}$ is type (2), there exists a reconfiguration sequence from $T_i$ to $T_{i+1}$ of $O(|V|)$ length \Cref{lem:suf:common}.
For two arborescences with the same root, there exists a reconfiguration sequence of length $O(|V|)$ by \Cref{thm:fixed-root}.
Concatenating the sequences, we obtain a reconfiguration sequence of length $O(|V|^2)$.
Therefore, we obtain the second part of \Cref{thm:dtr}.

Let us note that for yes-instances, the upper bound $O(|V|^2)$ on the length of reconfiguration sequences is tight up to a constant factor.
\Cref{fig:LB-example} illustrates an instance that requires to transform one tree into the other with $\Omega(|V|^2)$ steps.
This can be seen as follows.
Let the gray tree be $\source{T}$ and let the dashed tree be $\sink{T}$.
We first observe that every arborescence with $k + 1$ arcs must have $a_i$ for some $i$ as its root.
This implies that $\mathcal G$ contains $k$ vertices corresponding to $a_i$ for $1 \le i \le k$.
Since there is no directed path from $a_i$ to $a_j$ with $i \neq j$,
$a_i$ and $a_j$ is adjacent in $\mathcal G$ if and only if $|i - j| \le 1$.
Now, in order to transform an $a_i$-arborescence $T$ into an $a_{i+1}$-arborescence $T'$ with a single step, $T$ must contain $(a_i, b_{i+1})$ and $(b_{i+1}, c_j)$ for all $1 \le j \le k$.
Thus, from $\source{T}$, we need to transform it into such an $a_1$-arborescence with $k + 1$ steps, and then obtain an $a_2$-arborescence $T'$ with $k + 2$ steps in total.
By inductively applying this argument to each $1 \le i \le k$, the entire reconfiguration sequence requires $k(k + 2) = \Omega(|V|^2)$ steps in total.
\begin{figure}[t]
    \centering
    \includegraphics[width=0.5\textwidth]{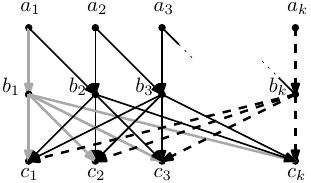}
    \caption{An example of requiring $\Omega(|V|^2)$ steps to transform the gray trees into the dashed trees.}
    \label{fig:LB-example}
\end{figure}

\section{Intractable Cases}\label{sec:hard}

In this section, we show negative results for some graph properties $\Pi$.
We will prove that when $\Pi$ is the property of being a directed path or a directed acyclic graph, the problem is PSPACE-complete, and when $\Pi$ is the property of being a strongly connected graph, the problem is NP-hard.

\subsection{Directed paths}\label{sec:path}
In this subsection, we show that {\sc Directed Path Reconfiguration} and {\sc Directed Path Sliding}, which are defined below, are both PSPACE-complete.
Thanks to the PSPACE-completeness of the undirected counterpart of {\sc Directed Path Sliding}~\cite{DEHJLUU19}, a very simple reduction shows that {\sc Directed Path Sliding} is PSPACE-complete as well (\Cref{thm:pspace-comp-path}).
In the following, we show that {\sc Directed Path Reconfiguration} is equivalent to {\sc Directed Path Sliding} in the complexity perspective.

{\sc Directed Path Reconfiguration} is a variant of {\sc Arborescence Reconfiguration}, where the two input trees $\source{T}$, $\sink{T}$ and intermediate trees are all directed paths in $G$.
Here, we use $\langle P_0, P_1, \ldots, P_\ell \rangle$ with $P_0 = \source{P}$ and $P_\ell = \sink{P}$ to denote a reconfiguration sequence between two directed paths $\source{P}$ and $\sink{P}$.
{\sc Directed Path Sliding} consists of the same instance of {\sc Directed Path Reconfiguration} and we are allowed the following adjacency relation in a valid reconfiguration sequence: for every pair of consecutive directed paths $P = (v_1, v_2, \ldots, v_k)$ and $P' = (v'_1, v'_2, \ldots, v'_k)$, either $v_i = v'_{i + 1}$ holds for all $1 \le i < k$ or $v_i = v'_{i - 1}$ holds for all $1 < i \le k$.
Since $P'$ is obtained by ``sliding'' in a forward or backward direction, we call the problem {\sc Directed Path Sliding}.
In this subsection, we show that {\sc Directed Path Reconfiguration} and {\sc Directed Path Sliding} are both PSPACE-complete.

To this end, we first show that both problems are equivalent with respect to polynomial-time many-one reductions.
Let $G$ be a directed graph and let $P = (v_1, v_2, \ldots, v_k)$ be a directed path in $G$ with arc $e_i = (v_i, v_{i+1})$ for $1 \le i < k$.
We denote by $t(P)$ the tail $v_1$ of $P$ and by $h(P)$ the head $v_k$ of $P$.
Observe that for a directed path $P'$ in $G$ with $|A(P) \setminus A(P')| = |A(P') \setminus A(P)| = 1$, at least one of the following conditions hold:
\begin{itemize}
    \item {\bf sliding}: $P' = (v_2, v_3, \ldots v_k, v)$ or $P' = (v, v_1, v_2, \ldots, v_{k - 1})$ for some $v \in V \setminus V(P)$;
    \item {\bf turning}: $P' = (v_1, v_2, \ldots, v_{k-1}, v)$ or $P' = (v, v_2, v_3, \ldots, v_k)$ for some $v \in V \setminus V(P)$;
    \item {\bf shifting}: $P' = (v_i, v_{i+1}, \ldots, v_k, v_1, \ldots, v_{i-1})$ for some $1 < i \le k$. This can be done when $P + (v_k, v_1)$ forms a directed cycle.
\end{itemize}
See Fig.~\ref{fig:op} for an illustration.

\begin{figure}[t]
    \centering
    \includegraphics{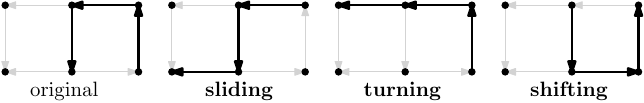}
    \caption{An illustration of the three operations in {\sc Directed Path Reconfiguration}.}
    \label{fig:op}
\end{figure}

We can regard these conditions as operations to obtain $P'$ from $P$. 
Since {\bf shifting} can be simulated by $i - 1$ {\bf sliding} operations along the directed cycle $P + (v_k, v_1)$, the essential difference between {\sc Directed Path Reconfiguration} and {\sc Directed Path Sliding} is the {\bf turning} operation in order to solve these problems.
Now, we perform polynomial-time reductions between these problems in both directions.

Let $(G = (V, A), \source{P}, \sink{P})$ be an instance of {\sc Directed Path Reconfiguration}.
For each vertex $v$ in $G$, we add two vertices $v^{\rm in}, v^{\rm out}$ and two arcs $(v^{\rm in}, v), (v, v^{\rm out})$.
These two vertices are called \emph{pendant vertices}.
We let $G'$ be the graph obtained in this way.
Then, we show the following lemma.

\begin{lemma}\label{lem:reconf-to-slide}
    $(G, \source{P}, \sink{P})$ is a yes-instance of {\sc Directed Path Reconfiguration} if and only if $(G', \source{P}, \sink{P})$ is a yes-instance of {\sc Directed Path Sliding}.
\end{lemma}
\begin{proof}
    Let $\langle P_0, P_1, \ldots, P_\ell \rangle$ be a reconfiguration sequence between $\source{P} = P_0$ and $\sink{P} = P_\ell$ of {\sc Directed Path Reconfiguration}.
    By the above argument, we can assume that $P_{i + 1}$ is obtained from $P_i$ by applying either {\bf sliding} or {\bf turning}.
    Let $P_i = (v_1, v_2, \ldots, v_k)$.
    We replace the subsequence $\langle P_i, P_{i + 1} \rangle$ with $\langle P_i, P', P_{i + 1} \rangle$, where $P' = (v_1^{\rm in}, v_1, v_2, \ldots, v_{k-1})$ if $t(P_i) = t(P_{i + 1})$ and $P' = (v_2, v_3, \ldots, v_{k}, v_k^{\rm out})$ otherwise.
    Clearly, $P'$ and $P_{i + 1}$ are obtained from $P_i$ and $P'$ by applying {\bf sliding} operations, respectively.
    By replacing each subsequence for $0 \le i < \ell$, we have a reconfiguration sequence of {\sc Directed Path Sliding} in $G'$.
    
    Conversely, let $\langle P_0, P_1, \ldots, P_\ell \rangle$ be a reconfiguration sequence $\source{P} = P_0$ and $\sink{P} = P_\ell$ of {\sc Directed Path Sliding}.
    Similarly to the other direction, we construct a reconfiguration sequence of {\sc Directed Path Reconfiguration}.
    Assume that $\source{P} \neq \sink{P}$ as otherwise we are done.
    Observe that each path $P_i = (v_1, v_2, \ldots, v_k)$ contains at most one pendant vertex.
    This follows from the fact that if $P_i$ contains both $v^{\rm in}$ and $w^{\rm out}$ for some $v, w \in V$, then $P_i$ cannot move to a distinct position by {\bf sliding} operations.
    Now, suppose $P_i$ is a directed path in $G$ with $P_i \neq \sink{P}$, that is, it has no pendant vertices.
    As $\sink{P}$ has no pendant vertices, we can find the smallest index $j > i$ such that $P_j$ has no pendant vertices.
    Since $P_j$ can be obtained from $P_i$ by {\bf sliding} or {\bf turning}, we can construct a reconfiguration sequence of {\sc Directed Path Reconfiguration} by omitting paths having pendant vertices.
\end{proof}

For the converse direction, we let $(G, \source{P}, \sink{P})$ be an instance of {\sc Directed Path Sliding}.
Let $G'$ be the directed graph obtained from $G$ by subdividing each arc $e = (u,w)$ with a new vertex $v_e$, that is, we replace $e$ with $v_e$ and add two arcs $(u,v_e)$ and $(v_e,w)$.
Let $\source{Q}$ and $\sink{Q}$ be defined accordingly from $\source{P}$ and $\sink{P}$, respectively.
In $G'$, we say that a path $P'$ is a \emph{standard path} if $h(P')$ and $t(P')$ belong to $V$ and it is a \emph{nonstandard path} otherwise.

\begin{lemma}\label{lem:slide-to-reconf}
    $(G, \source{P}, \sink{P})$ is a yes-instance of {\sc Directed Path Sliding} if and only if $(G', \source{Q}, \sink{Q})$ is a yes-instance of {\sc Directed Path Reconfiguration}.
\end{lemma}

\begin{proof}
    It is easy to transform any reconfiguration sequence of $(G, \source{P}, \sink{P})$ for {\sc Directed Path Sliding} to that of $(G', \source{Q}, \sink{Q})$ for {\sc Directed Path Reconfiguration}.
    Conversely, let $\langle Q_0, Q_1, \ldots, Q_\ell \rangle$ be a reconfiguration sequence of {\sc Directed Path Reconfiguration} between $\source{Q} = Q_0$ and $\sink{Q} = Q_\ell$ in $G'$.
    Observe that {\bf turning} is allowed only for nonstandard paths.
    This means that for any two standard paths $Q_i$ and $Q_j$ in a reconfiguration sequence such that $Q_k$ is nonstandard for $i < k < j$, $Q_j$ is obtained from $Q_i$ by two {\bf sliding} operations.
    Thus, by replacing each subsequence $\langle Q_i, Q_{i + 1}, \ldots,  Q_j \rangle$ in this way, we obtain that of $(G', \source{Q}, \sink{Q})$ for {\sc Directed Path Sliding}, which also gives a reconfiguration sequence of $(G, \source{P}, \sink{P})$ for {\sc Directed Path Sliding} as well.
\end{proof}

Now, we show the PSPACE-completeness of {\sc Directed Path Sliding}.
\begin{theorem}\label{thm:pspace-comp-path}
    {\sc Directed Path Sliding} is PSPACE-complete.
\end{theorem}

\begin{proof}
    By a standard argument in reconfiguration problems, the problem belongs to PSPACE: By non-deterministically guessing the ``next solution'' in a reconfiguration sequence, the problem can be solved in non-deterministic polynomial space, while by Savitch's theorem~\cite{Sav70}, we can solve the problem in deterministic polynomial space as well.
    
    It is easy to observe that the undirected version of {\sc Directed Path Sliding} can be reduced to {\sc Directed Path Sliding} by simply replacing each (undirected) edge of an input graph with two arcs with opposite directions.
    As the undirected version is known to be PSPACE-complete~\cite{DEHJLUU19}, the directed version is also PSPACE-complete.
\end{proof}

By~\Cref{lem:slide-to-reconf}, we immediately have the following corollary.
\begin{corollary}
    {\sc Directed Path Reconfiguration} is PSPACE-complete.
\end{corollary}

\subsection{Directed acyclic graphs}\label{sec:dag}
Suppose that subgraphs in a reconfiguration sequence are relaxed to be acyclic.
Observe that the problem is equivalent to reconfiguring directed feedback arc sets in directed graphs.
More specifically, given two directed acyclic subgraphs $\source{H}$ and $\sink{H}$ in a directed graph $G = (V, A)$, the problem asks to determine whether there is a reconfiguration sequence of directed acyclic subgraphs $\langle \source{H} = H_0, H_1, \ldots, H_\ell = \sink{H} \rangle$ such that $|A(H_i) \setminus A(H_{i + 1})| = |A(H_{i + 1}) \setminus A(H_i)| = 1$ for all $0 \le i < \ell$.
Seeing this problem from the complement, the problem is equivalent to finding a reconfiguration sequence $\langle A_1, A_2, \ldots, A_\ell\rangle$ of subsets of $A$ such that $H_i = G - A_i$ is acyclic for all $0 \le i \le \ell$.
Since each $A_i$ is a feedback arc set of $G$, we call this problem {\sc Directed Feedback Arc Set Reconfiguration}.
There is another variant of this problem, called {\sc Directed Feedback Vertex Set Reconfiguration}, in which we are asked to determine two given subsets $\source{V}$ and $\sink{V}$ of $V$, there is a sequence of vertex subsets $\langle \source{V} = V_0, V_1, \ldots, V_\ell = \sink{V} \rangle$ of $V$ such that $G[V \setminus V_i]$ is acyclic and $|V_i \setminus V_{i + 1}| = |V_{i + 1} \setminus V_i| = 1$ for all $0 \le i < \ell$.

\begin{theorem}
    {\sc Directed Feedback Arc Set Reconfiguration} and {\sc Directed Feedback Vertex Set Reconfiguration} are PSPACE-complete.
\end{theorem}

\begin{proof}
By an analogous argument in \Cref{thm:pspace-comp-path}, these problems belong to PSPACE.

It is easy to observe that {\sc Directed Feedback Vertex Set Reconfiguration} is PSPACE-hard.
To see this, consider an undirected graph $G = (V, E)$ and the directed graph $D = (V, A)$ obtained from $G$ by replacing all undirected edge $\{u, v\}$ with two arcs $(u, v)$ and $(v, u)$.
Observe that every vertex cover of $G$ is also a directed feedback vertex set of $D$ and vice versa.
By the PSPACE-hardness of reconfiguring independent sets~\cite{IDHPSUU11}, {\sc Directed Feedback Vertex Set Reconfiguration} is PSPACE-hard.

To prove the PSPACE-hardness of {\sc Directed Feedback Arc Set Reconfiguration},
we perform a standard polynomial-time reduction from
{\sc Directed Feedback Vertex Set Reconfiguration}.

Let $G = (V, A)$ be a directed graph.
We construct a directed multigraph $G' = (V', A')$ as follows.
We first add a pair of copies $\{v^{\rm in}, v^{\rm out}\}$ for each $v \in V$ and add an arc $(v^{\rm in}, v^{\rm out})$ to $G'$.
We call this arc an \emph{internal arc} of $v$.
The vertex set of $G'$ is defined as $V' = \bigcup_{v \in V} \{v^{\rm in}, v^{\rm out}\}$.
For $(u, v) \in A$, add $|V| + 1$ parallel arcs $(u^{\rm out}, v^{\rm in})$ to $G'$.
For two (directed) feedback vertex sets $\source{X}$ and $\sink{X}$ in $G$ with $|\source{X}| = |\sink{X}| = k$, $\source{Y}$ and $\sink{Y}$ defined as the sets of internal arcs corresponding to $\source{X}$ and $\sink{X}$, respectively.
Now, we show that $G$ contains a reconfiguration sequence of feedback vertex sets between $\source{X}$ and $\sink{X}$ in $G$ if and only if there is a reconfiguration sequence of (directed) feedback arc sets between $\source{Y}$ and $\sink{Y}$ in $G'$.

Since $X$ is a feedback vertex set of $G$, the corresponding internal arc set $Y$ is a feedback arc set of $G'$.
Thus, the forward implication is straightforward.
Conversely, suppose that there is a reconfiguration sequence $\langle Y_0, Y_1, \ldots, Y_\ell\rangle$ between $Y_1 = \source{Y}$ and $Y_\ell = \sink{Y}$ such that all the intermediate sets $Y_i$ are feedback arc sets of $G'$.
For $0 \le i \le \ell$, let $Y'_i$ be the set of internal arcs in $Y_i$ and $X'_i$ be the set of vertices in $G$, each of which corresponds to an (internal) arc in $Y'_i$.
To prove the backward implication, it suffices to show that $X'_i$ is a feedback vertex set of $G'$.
To see this, suppose that there is a directed cycle $C$ in $G[V \setminus X'_i]$.
For every arc $(u, v)$ in $C$, there is at least one arc from $u^{\rm out}$ to $v^{\rm in}$ in $G' - Y_i$ as there are $|V| + 1$ copies there.
Thus, the cycle also induces a directed cycle in $G' - Y_i$, contradicting the fact that $Y_i$ is a feedback arc set of $G'$.
\end{proof}

\subsection{Strongly connected graphs}\label{sec:strong}
In \Cref{sec:path,sec:dag}, we have considered acyclic properties $\Pi$.
As another direction, we consider the case where $\Pi$ is the property of being strongly connected in this subsection.
A directed graph is \emph{strongly connected} if for any two vertices $u$ and $v$, the graph contains directed paths from $u$ to $v$ and from $v$ to $u$.
We consider two variants: {\sc Strongly Connected Vertex Set Reconfiguration} and {\sc Strongly Connected Arc Set Reconfiguration}.
In the vertex variant, we are given two subsets $\source{V}$ and $\sink{V}$ of $V(G)$ and asked whether there is a sequence of subsets $\langle \source{V} = V_0, V_1, \ldots, V_\ell = \sink{V} \rangle$ of $V(G)$ such that $G[V_i]$ is strongly connected for all $0 \le i \le \ell$ and $|V_i \setminus V_{i + 1}| = |V_{i + 1} \setminus V_i| = 1$ for all $0 \le i < \ell$.
The arc variant is defined in an analogous way for arc subsets: We are given two subsets $\source{A}$ and $\sink{A}$ of $A(G)$ and asked whether there is a sequence of subsets $\langle \source{A} = A_0, A_1, \ldots, A_\ell = \sink{A} \rangle$ of $A(G)$ such that the subgraph $G[A_i]$ induced by $A_i$ forms is strongly connected for all $0 \le i \le \ell$ and $|A_i \setminus A_{i + 1}| = |A_{i + 1} \setminus A_i| = 1$ for all $0 \le i < \ell$.
In this subsection, we will show that the vertex variant is PSPACE-complete, and the arc variant is NP-hard for oriented graphs.

For the vertex variant, we show a reduction from {\sc Shortest Path Reconfiguration}, which is known to be PSPACE-complete~\cite{Bonsma13}.
Our reduction is similar to that of Hanaka et al.~\cite{HIMMNSSV20} for {\sc Induced Path Reconfiguration} in undirected graphs.
In {\sc Shortest Path Reconfiguration}, we are given a simple undirected graph $G$, with specified vertices $p$ and $q$, and two subsets $\source{V}$ and $\sink{V}$ of $V(G)$ that are (induced) shortest $p$-$q$ paths in $G$.
The question is whether there exists a sequence of vertex subsets $\langle \source{V} = V_0, V_1, \ldots, V_\ell = \sink{V} \rangle$ of $V(G)$ such that $V_i$ is a shortest $p$-$q$ path for all $0 \le i \le \ell$ and $|V_i \setminus V_{i + 1}| = |V_{i + 1} \setminus V_i| = 1$ for all $0 \le i < \ell$.

Let $d$ be the length of a shortest $p$-$q$ path in $G$.
For $i \in \{0, 1, \dots, d\}$, we denote by $L_i \subseteq V(G)$ the set of vertices such that the distance from $p$ is $i$ and that to $q$ is $d - i$.
It follows that $L_0 = \{p\}$ and $L_d = \{q\}$.
We call each $L_i$ a \emph{layer}.
Observe that every shortest $p$-$q$ path contains exactly one vertex from each layer.
By this observation, we can assume without loss of generality that every vertex in $G$ belongs to some layer, and every edge of $G$ joins vertices in adjacent layers, that is, for every edge $\{u, v\}$ there exists $i \in \{0, 1, \dots, d-1\}$ such that $u \in L_i$ and $v \in L_{i+1}$.

\begin{theorem}\label{thm:vertex-strong}
    {\sc Strongly Connected Vertex Set Reconfiguration} is PSPACE-complete.
\end{theorem}
\begin{proof}
    By an analogous argument in \Cref{thm:pspace-comp-path}, the problem belongs to PSPACE.
    
    Given an instance $I = (G, p, q, \source{V}, \sink{V})$ of {\sc Shortest Path Reconfiguration}, we construct a directed graph $D$ by orienting the edges in $G$ from $L_i$ to $L_{i+1}$ for every $i \in \{0, 1, \dots, d-1\}$ and adding an arc $(q, p)$, which can be done in polynomial time.
    In $D$, both $D[\source{V}]$ and $D[\sink{V}]$ are strongly connected because they induce directed cycles.
    Moreover, $V' \subseteq V(G) (= V(D))$ induces a shortest $p$-$q$ path in $G$ if and only if $V'$ induces a directed cycle in $D$.
    Therefore, $I$ is a yes-instance if and only if $(D, \source{V}, \sink{V})$ is a yes-instance of {\sc Strongly Connected Vertex Set Reconfiguration}.
\end{proof}

Next, we show that the arc variant is NP-hard.
In this variant, we can assume that the given two arc sets are spanning without loss of generality.
This is because of the following reasons:
If there exist two adjacent arc sets $A_i$ and $A_{i+1}$ with $v \in V(G[A_i])$ and $v \notin V(G[A_{i+1}])$ in a reconfiguration sequence, $A_i$ contains a single arc whose endpoint is $v$, indicating that $G[A_i]$ is not strongly connected.
Symmetrically, there exist no two adjacent arc sets with $v \notin V(G[A_i])$ and $v \in V(G[A_{i+1}])$.

\renewcommand\vec[1]{\overrightarrow{#1}}
\newcommand\cev[1]{\overleftarrow{#1}}

\begin{theorem}
    {\sc Strongly Connected Arc Set Reconfiguration} is NP-hard for oriented graphs.
\end{theorem}
\begin{proof}
    We show a reduction from {\sc Directed Hamiltonian Cycle}, which is NP-complete even for oriented graphs~\cite{Plesnik79}.
    The input of the problem is a directed graph $G = (V, A)$ with $n = |V(G)|$ vertices and $m = |A(G)|$ arcs, and the question is whether $G$ contains a Hamiltonian cycle, that is, a spanning directed cycle.
    We assume that $m \ge n$. (Otherwise, the problem is trivial.)
    Given $G$, we construct a directed graph $H$ as follows.
    For $k = 2(m-n) + 3$, we define a directed graph $D_k$ with the vertex set $V(D_k) = [k]$, where $[k] = \{1, 2, \dots, k\}$.
    For convenience, the addition $+$ and subtraction $-$ are taken over modulo $k$,
    i.e., $k + 1$ is regarded as $1$ and $0$ is regarded as $k$.
    We define $\vec{e_i} = (i, i + 1)$, $\cev{e_i} = (i, i - 2)$, $\vec{C_k} = ([k], \{\vec{e_i} \colon i \in [k]\})$, and $\cev{C_k} = ([k], \{\cev{e_i} \colon i \in [k]\})$.
    Both $\vec{C_k}$ and $\cev{C_k}$ form directed cycles because $k$ is odd. 
    The arc set $A(D_k)$ is the disjoint union of $A(\vec{C_k})$ and $A(\cev{C_k})$.
    Given $G$, we define $H$ as the graph obtained by choosing arbitrary one vertex from each $G$ and $D_k$ and then identifying them.
    In the following, we write simply $\vec{C}$ and $\cev{C}$ to denote $A(\vec{C_k})$ and $A(\cev{C_{k}})$, respectively.
    Let $S = A(G) \cup \vec{C}$ and $S' = A(G) \cup \cev{C}$.
    We show that $G$ is a yes-instance of {\sc Directed Hamiltonian Cycle} if and only if $(H, S, S')$ is a yes-instance of {\sc Strongly Connected Arc Set Reconfiguration}.
    
    Suppose that $G$ is a yes-instance of {\sc Directed Hamiltonian Cycle}.
    Let $B \subseteq A(G)$ be the arc set of a Hamiltonian cycle in $G$.
    Then, we can reconfigure $S$ into $S'$ in $H$ as follows.
    First, we move $m - n$ arcs in $A(G) \setminus B$ to $\cev{C}$ so that they form a path.
    Next, we add an arc $\cev{e}$ in $\cev{C}$ so that the path is extended in the forward direction and then remove the arc in $\vec{C}$ directed to the head of $\cev{e}$.
    We repeat this procedure $m - n + 3$ times.
    In each step, the arc set is strongly connected because the last $(|\cev{C}| - 1) / 2$ arcs in the extended path starts at the tail of the removed arc and ends at its head.
    Finally, we move the remaining $m - n$ arcs in $\vec{C}$ to $A(G) \setminus B$ and obtain $S'$.
    
    \begin{figure}[t]
        \centering
        \includegraphics[width=0.6\linewidth]{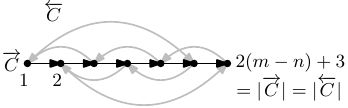}
        \caption{An illustration for the proof of \Cref{clm:cycle}.}
        \label{fig:sc_edge_oriented}
    \end{figure}
    
    Suppose that $(H, S, S')$ is a yes-instance of {\sc Strongly Connected Arc Set Reconfiguration}.
    Let $\langle S = S_0, S_1, \dots, S_{\ell} = S' \rangle$ be a reconfiguration sequence and $i$ be the integer such that $\vec{C} \subseteq S_j$ for $0 \le j \le i$ and $\vec{C} \not\subseteq S_{i+1}$.
    \let\origqedsymbol\qedsymbol
    \renewcommand{\qedsymbol}{$\vartriangleleft$}
    \begin{claim}\label{clm:cycle}
        $|S_i \cap \cev{C}| \ge \frac{|\cev{C}| - 3}{2}\ (= m - n)$.
    \end{claim}
    \begin{proof}
        If $|S_i \cap \cev{C}| < \frac{|\cev{C}| - 3}{2}$, then $|S_{i+1} \cap \cev{C}| \le \frac{|\cev{C}| - 3}{2}$.
        Since $|S_{i+1} \cap \vec{C}| = |\vec{C}| - 1$, $S_{i+1} \cap (\vec{C} \cup \cev{C})$ is a disjoint union of a directed path consisting of $|\vec{C}| - 1$ arcs from $\vec{C}$ and at most $\frac{|\cev{C}| - 3}{2}$ arcs from $\cev{C}$.
        In other words, $S_{i+1} \cap (\vec{C} \cup \cev{C})$ is an arc set obtained from the graph consisting of black arcs in \Cref{fig:sc_edge_oriented} by adding at most $\frac{|\cev{C}| - 3}{2}$ gray arcs from $\cev{C}$ to make the graph strongly connected.
        In the figure, the vertices are renamed so that the black path starts at $1$ and ends at $|\vec{C}|$.
        To make the graph strongly connected, it is necessary (and sufficient) that by adding arcs the resultant graph has a path from $|\vec{C}|$ to $1$.
        However, every path from $|\vec{C}|$ to $1$ in $(\vec{C} \cup \cev{C}) \setminus \{(|\vec{C}|, 1)\}$ uses at least $\frac{|\cev{C}| - 1}{2}$ arcs from $\cev{C}$.
    \end{proof}
    \renewcommand{\qedsymbol}{\origqedsymbol}
    By \Cref{clm:cycle}, it follows that $|S_i \cap A(G)| \le n$ because
    \begin{align*}
        |S_i \cap A(G)|
        &= |S_i| - |S_i \cap \vec{C}| - |S_i \cap \cev{C}|\\
        &\le |S_i| - |\vec{C}| - \frac{|\cev{C}| - 3}{2}\\
        &= (m + 2(m - n) + 3) - (2(m - n) + 3) - (m - n) = n.
    \end{align*}
    For $B = S_i \cap A(G)$, $G[B]$ is strongly connected because the common vertex of $G$ and $D_k$ is a cut vertex.
    Since $|B| \le n$ and $|V(G)| = n$, $B$ forms a Hamiltonian cycle in $G$.
\end{proof}

\section{Concluding Remarks}

There are several possible open questions related to our results.
Contrary to the cases of spanning arborescences and spanning $r$-arborescences, the sets of arborescences and $r$-arborescences with $k < |V| - 1$ arcs do not satisfy the weak exchange property, which makes {\sc Arborescence Shortest Reconfiguration} highly nontrivial.
{\sc Arborescence Shortest Reconfiguration} would be a notable open question arising in our work.
It would be also interesting to know whether {\sc Directed Path Reconfiguration} and {\sc Directed Path Sliding} are fixed-parameter tractable (FPT) when parameterized by the length of input paths.
Although the undirected counterparts are known to be FPT~\cite{DEHJLUU19,GSZ20}, it would be difficult to apply their techniques directly to our cases.
Another question is whether \textsc{Strongly Connected Arc Set Reconfiguration} belongs to NP or is PSPACE-complete.
We have shown that the problem is NP-hard, while the vertex variant is PSPACE-complete.

\section*{Acknowledgment}
We thank Anna Lubiw and one of the reviewers for pointing out the work of Lov\'{a}sz~\cite{lovasz1977homology}, which is related to our result.
We also thank the reviewers for helpful comments.

\printbibliography
\end{document}

%% file: macro.tex
\usepackage[linesnumbered,ruled,vlined]{algorithm2e}
\SetKwInOut{Input}{input}
\SetKwInOut{Output}{output}
\usepackage{amsmath}
\usepackage{amssymb}
\usepackage{amsfonts}
\usepackage{amsthm}
\usepackage{cleveref}
\usepackage{stmaryrd}
\usepackage{setspace}
\usepackage{tabularx}
\usepackage{url}
\usepackage[subrefformat=parens]{subcaption}
\usepackage{tikz}
\usetikzlibrary{
  positioning,
  calc,
}
\usepackage{lipsum}
\usepackage[shortlabels]{enumitem}
\usepackage{cleveref}
\usepackage{amsmath,amssymb,amsthm}
\usepackage{mathtools}

\newcommand{\source}[1]{{{#1}^{\mathsf s}}}
\newcommand{\sink}[1]{{#1}^{\mathsf t}}

\newtheorem{theorem}{Theorem}
\newtheorem{lemma}{Lemma}
\newtheorem{corollary}{Corollary}

\newtheorem{claim}{Claim}

\theoremstyle{definition}
\newtheorem{remark}{Remark}

